\documentclass{article}
\usepackage{amsmath,amssymb,amsthm, graphicx}
\usepackage[dvips]{epsfig}
\usepackage{amsfonts}
\usepackage{amscd}
\usepackage{color}

\newtheorem{theorem}{Theorem}[section]
\newtheorem{proposition}[theorem]{Proposition}
\newtheorem{corollary}[theorem]{Corollary}
\newtheorem{lemma}[theorem]{Lemma}
\newtheorem{claim}[theorem]{Claim}

\newtheorem{question}[theorem]{Question}
\newtheorem{conjecture}[theorem]{Conjecture}
\theoremstyle{definition}
\newtheorem{definition}[theorem]{Definition}
\theoremstyle{remark}

\newtheorem{remark}[theorem]{Remark}

\newcommand{\R}{\mathbb{R}}
\newcommand{\Z}{\mathbb{Z}}

\newcommand{\N}{{\mathbb N}}

\newcommand{\F}{\mathcal{F}}
\newcommand{\G}{\mathcal{G}}

\newcommand{\pr}{\mathrm{pr}}

\newcommand{\heading}[1]{\medskip\par\noindent{\bf #1}}

\newcommand{\BA}{\mathbf{A}}
\newcommand{\BB}{\mathbf{B}}

\newcommand{\BU}{\mathbf{U}}
\newcommand{\BV}{\mathbf{V}}

\newcommand{\BX}{\mathbf{X}}

\newcommand{\red}[1]{{\color{red}#1}}
\newcommand{\blue}[1]{{\color{blue}#1}}
\definecolor{orange}{rgb}{0.9,0.45,0}
\newcommand{\oldlemma}[1]{}

\DeclareMathOperator{\conv}{conv}
\DeclareMathOperator{\interior}{Int}

\DeclareMathOperator{\st}{st}
\DeclareMathOperator{\sd}{sd}
\DeclareMathOperator{\cone}{Con}

\begin{document}

\title{Nerves of good covers are algorithmically unrecognizable}

\author{Martin Tancer\thanks{Department
of Applied Mathematics,
Charles University in Prague, Malostransk\'{e} n\'{a}m.
25, 118~00~~Praha~1
\emph{and}
Institutionen f\"{o}r matematik, Kungliga Tekniska H\"{o}gskolan, 100~44
Stockholm. 
Partially supported by the ERC Advanced Grant No.
267165 and by the Center of Excellence -- Inst.\ for Theor.\
Comp.\ Sci., Prague (project P202/12/G061 of GA~\v{C}R).
}
\and 
Dmitry Tonkonog\thanks{
Department of Differential Geometry and Applications, Faculty of Mechanics and Mathematics,
Moscow State University, GSP-1, Leninskie Gory, Moscow 119991, Russia,
\emph{and}
Delone Laboratory of Discrete and
Computational Geometry, Yaroslavl State University, 
Sovetskaya str.~14, Yaroslavl 150000, Russia. 
Partially supported by the Russian Government project 11.G34.31.0053,
RFBR grant 12-01-00748-a,
and by the Dobrushin Scholarship at the Independent University of Moscow.
}
}

\maketitle

\begin{abstract}
A \emph{good cover} in $\R^d$ is a collection of open contractible sets in
$\R^d$ such that the intersection of any subcollection is either contractible or
empty. Motivated by an analogy with convex sets, intersection patterns of good
covers were studied intensively. Our main result is that
intersection patterns of good covers are algorithmically unrecognizable.

More precisely, the intersection pattern of a good cover can be stored
in a simplicial complex called \emph{nerve} which records which subfamilies of
the good cover intersect. A simplicial complex is \emph{topologically
$d$-representable} if it is isomorphic to the nerve of a good cover in $\R^d$.
We prove that it is algorithmically undecidable
whether a given simplicial complex is
 topologically $d$-representable 
for any fixed $d \geq 5$. 
The result remains also valid if we replace
good covers with acyclic covers or with covers by open $d$-balls. 

As an auxiliary result we prove that if a simplicial complex is PL embeddable
into $\R^d$ then it is topologically $d$-representable. We also supply this
result with showing that if a ``sufficiently fine'' subdivision of a $k$-dimensional complex is $d$-representable and $k \leq \frac{2d-3}3$, then the complex is PL embeddable into $\R^d$.

\end{abstract}

\section{Introduction}
Many results in discrete geometry are devoted to studying intersection patterns
of convex sets. A pioneering result in this respect is the Helly
theorem~\cite{helly23}. It states that whenever $C_1, \dots, C_n$
are convex sets in $\R^d$, $n \geq d+1$, such that the intersection of any
$d+1$ of these sets is nonempty, then the intersection of all sets is
nonempty. Many results of similar flavor are known and the interested reader is
referred to the survey paper~\cite{tancer11surveyarxiv} for more details.

\heading{Nerves and $d$-representable complexes.}
For a collection of sets,
its intersection pattern 
can be 
encoded into a combinatorial object that
is called
the nerve of the collection.

Consider a collection of sets $\F=\{F_1, \dots,
F_n\}$. 
The \emph{nerve} of $\F$ is the simplicial
complex\footnote{We briefly
recall simplicial complexes
and related definitions in
section~\ref{s:pl}.

} whose $k$-dimensional
faces are the subcollections $\{F_{i_1},\ldots,F_{i_k}\}$ of $\F$ 
such that $F_{i_1}\cap\ldots\cap F_{i_k}\neq\emptyset$.
In particular, the nerve of $\F$
has $n$ vertices
$F_1, \dots, F_n$ (provided that $F_i$ are nonempty).

\begin{definition}
A \emph{convex cover} in $\R^d$
is a finite collection of
open convex sets $U_1,\ldots,U_n\subset \R^d$.
\end{definition}

\begin{remark}
Note that we do not require
$\bigcup_{i=1}^n U_i=\R^d$.
The word `cover' should not be misleading.
\end{remark}
 
\begin{definition}
A simplicial complex is \emph{$d$-representable} if it is
isomorphic 
to the nerve of a convex cover in $\R^d$. 
\end{definition}

\heading{Topological $d$-representability.}
The following
generalization of a 
convex cover is rather well-known
in topology.

\begin{definition}
A \emph{good cover} in $\R^d$ is a finite collection of
open sets $U_1,\ldots,U_n$ in $\R^d$ such that 
the intersection
$U_{i_1}\cap\ldots\cap U_{i_k}$
of any (nonempty) subcollection
$\{U_{i_1},\ldots,U_{i_k}\}$ is either empty 
or contractible. (In particular,
$U_i$ are contractible.)
\end{definition}

\begin{remark}
A convex cover 
is a good cover.
\end{remark}

\begin{definition}
\label{d:top_rep}
A simplicial complex is 
\emph{topologically $d$-representable}\footnote{
Topological $d$-representability was first introduced
in~\cite{tancer12counterex}. However, the definition was slightly different; see 
Definition~\ref{d:rep_balls}
.
}
if it is
isomorphic
to the nerve of a good cover in $\R^d$. 
\end{definition}

Classifying 
intersection patterns of 
convex (or good) covers in $\R^d$ is
equivalent to 
classifying $d$-representable 
(resp.~topologically $d$-representable) complexes.

Intersection patterns (formally, nerves)
of good covers 
inherit many properties of 
intersection patterns of
convex covers. For example, the Helly theorem
was generalized to good covers again by Helly~\cite{helly30}.
Another example is the well-known Nerve theorem,
see Theorem~\ref{t:nerve} below. (Probably, this theorem is the main
reason that makes good covers easier to study than arbitrary
collections of non-convex sets in $\R^d$.) Other examples include topological
versions of various Helly-type theorems~\cite{alon-kalai-matousek-meshulam02,
kalai-meshulam05,kalai-meshulam08}.

The main result of our paper, Theorem~\ref{t:unrecog},
is in the opposite spirit.
We show that from the general algorithmic viewpoint,
intersection patterns of good covers
behave differently (in fact, much worse) than
intersection patterns of convex covers.

First, recall the following theorem.

\begin{theorem}[\cite{wegner67}, see also~\cite{tancer11surveyarxiv}]
\label{t:recog_convex}
It is algorithmically \emph{decidable} whether a given simplicial complex is
$d$-representable. (There is actually a PSPACE algorithm for recognition
$d$-representable simplicial complexes.)
\end{theorem}

As we will now show, the situation with
topological $d$-representability is
completely different.
The main goal of our paper is to prove the following result.

\begin{theorem}[main result]
\label{t:unrecog}
For each $d\ge 5$, it is algorithmically \emph{undecidable} whether a given simplicial complex is
topologically $d$-representable.
\end{theorem}


\heading{Algorithmically undecidable problems.}
A decision problem is the following question: given a finite
input string $s$ (over a finite alphabet), decide whether $s$ satisfies a
certain property $P$. Roughly speaking, a decision computational problem is
algorithmically undecidable if there is no algorithm that would solve this
problem for every string $s$. More precisely, there is no Turing machine solving
this problem. However, the reader is not assumed to have background in Turing
machines, since the details about Turing machines are hidden in reduction of
our problem to famous Novikov's theorem (see Theorem~\ref{t:novikov} below). 

Undecidable problems naturally appear in algebra. For instance the
word problem for
groups or semigroups is known to be undecidable~\cite{post47, novikov55}. 
These problems also reflect in topology. For example it is
algorithmically undecidable whether the fundamental group of a given complex is
trivial since the word problem reduces to triviality of the fundamental
group. Another example is the above-mentioned Novikov's theorem. 
Briefly, it states that for each $d\ge 5$ it is algorithmically undecidable
whether a given simplicial complex is homeomorphic to the $d$-sphere.

In combinatorial geometry, undecidable problems are not so frequent (in the
authors' opinion; depending on how broadly combinatorial
geometry is considered). We should mention Wang's tiling problem proved undecidable by
Berger~\cite{berger66} as an example. Our problem is actually on 
the borderline area between  topology and combinatorial geometry. 
We hope that our approach could
have consequences in other problems in combinatorial geometry.

\heading{Other types of covers.}
\begin{definition}
\label{d:rep_balls}
Let us call
a simplicial complex
{\it topologically $d$-representable by balls}
if it is a nerve of a good cover
in $\R^d$ such that the intersection of any (nonempty) subcollection,
unless it is empty,
is not only contractible, but even
homeomorphic to the open $d$-ball.
\end{definition}

In~\cite{tancer12counterex},
topological $d$-representability by balls
was introduced as 
topological $d$-representability (in order to get a stronger result with this
definition).
Definition~\ref{d:top_rep} of
topological $d$-representability that we use in this paper is probably
more standard in the literature,
see, e.g.,~\cite{alon-kalai-matousek-meshulam02, kalai-meshulam08}. (These papers do
not define topological $d$-representability; however, they actually prove some
properties of topologically $d$-representable complexes.)
For completeness we prove the following
modification of Theorem~\ref{t:unrecog}.
\begin{theorem}
\label{t:unrecog_balls}
For each $d\ge 5$, it is algorithmically undecidable whether a given simplicial complex is
topologically $d$-representable by balls.
\end{theorem}

Let us also state 
a similar undecidability theorem for another version of covers: acyclic covers.

\begin{definition}
\label{d:acyclic}
An \emph{acyclic cover} in $\R^d$ is a finite collection of
open sets $U_1,\ldots,U_n$ in $\R^d$ such that 
the intersection
$U_{i_1}\cap\ldots\cap U_{i_k}$
of any (nonempty) subcollection
$\{U_{i_1},\ldots,U_{i_k}\}$ is acyclic (i.e., is
empty or has homology of a ball). 
Let us call
a simplicial complex
{\it $d$-representable by an acyclic cover}
if it is a nerve of an acyclic cover in $\R^d$.
\end{definition}

\begin{theorem}
\label{t:unrecog_acyclic}
For each $d\ge 5$, it is algorithmically undecidable whether a given simplicial complex is
topologically $d$-representable by an acyclic cover.
\end{theorem}
Suggestion to investigate also acyclic covers is by Roman
Karasev~\cite{karasev:pc12}. 
There are two reasons why this theorem is worth stating. First, nerves of
acyclic covers (or nerves of families of limited homological complexity) 
are widely investigated since many Helly-type theorems are also valid in this
case; see~\cite{colindeverdiere-ginot-goaoc12, hell05arxiv, kalai-meshulam05}.
Second, acyclic covers behave better than good covers from the algorithmic
viewpoint. If we have a combinatorially defined collection of open sets in
$\R^d$ (say, given as interiors of  polyhedra), then there is no algorithm
deciding whether the cover is good (because of Novikov's theorem), but there is an
algorithm deciding whether the cover is acyclic (
computing homology is
algorithmic). This remark shows that unrecognizability of nerves of good covers
stated in Theorem~\ref{t:unrecog} is not much related to unrecognizability of
good covers themselves.

Although
$d$-representability by balls implies
topological $d$-representability which implies
$d$-representability by acyclic covers, 
there are no a priori implications between
Theorems~\ref{t:unrecog},~\ref{t:unrecog_balls} and~\ref{t:unrecog_acyclic}.
On the other hand, with our approach the proofs are similar. We prove Theorems~\ref{t:unrecog} and
Theorem~\ref{t:unrecog_balls} simultaneously in Section~\ref{sec:planofproof},
postponing the proofs of key results to Sections~\ref{s:etr},~\ref{s:fund}. We
prove Theorem~\ref{t:unrecog_acyclic} only in Appendix~\ref{app_acyclic} because it involves additional technical details.

\section{Preliminaries}
\label{s:pl}
In this section, we quickly recall some basic definitions
and notations mostly concerning simplicial complexes.
A reader new to the topic might also want to
see more substantial literature~\cite{matousek03, hatcher01,
munkres84, rourke-sanderson72}. We recommend to consult preliminaries only
if the need arises.

\heading{Integers.} For an integer $n$ the symbol $[n]$ denotes the set
$\{1, 2, \dots, n\}$. 

\heading{Abstract simplicial complexes.}
Let $V$ be a finite set.
A collection $K$ of subsets of $V$ is a simplicial
complex if, together with each $\alpha\in K$,
we have $\beta\in K$ for every $\beta\subset \alpha$.
Any $\sigma\in K$ with $\# \sigma=k+1$ is called
a $k$-dimensional simplex (or {\it face}) of $K$ (by $\# \sigma$ we mean the
number of elements of $\sigma$).
The set $V$ is the set of \emph{vertices} of $K$. Usually we denote it by
$V(K)$.

Let $U \subset V$. The induced subcomplex of $K$ on $U$ is given by 
$K[U] := \{\sigma \in K\colon \sigma \subseteq U\}$.

Let $K,L$ be two simplicial complexes.
A map $f\colon V(K)\to V(L)$ is a {\it simplicial map} if $f(\alpha) \in L$ for
every $\alpha \in K$.
Two complexes $K,L$ are said to be {\it isomorphic}
if there is a bijective simplicial map $V(K)\to V(L)$.

\heading{Geometric realizations.} 
We work a priori with abstract
simplicial complexes. However, sometimes it is
more convenient to work with geometrical
realizations of abstract simplicial complexes.
Given an abstract simplicial complex $K$, we chose a map $f \colon V(K) \to
\R^m$ for sufficiently large $m$.
Assume that $f$ satisfy the
following properties:
\begin{itemize}
\item The set $f(\alpha)$ is affinely independent for every $\alpha \in K$; and
\item the convex hulls satisfy the relation $\conv(f(\alpha)) \cap \conv(f(\beta)) = \conv(f(\alpha \cap \beta))$
for every $\alpha, \beta \in K$.
\end{itemize}
If $m$ is large enough, then such an $f$ exists. For example, a map
 sending vertices of $K$ injectively to the vertices of a (geometric) simplex
in $\R^m$ is a suitable choice.

For a face $\alpha \in K$ we have the \emph{geometric realization} of
this face
$$
|\alpha| := \conv\{f(v)\colon v \in \alpha\}.
$$
We also have the \emph{geometric
realization} of any subcomplex $X\subset K$ given by
$$
|X| := \bigcup_{\alpha \in X}|\alpha|.
$$ 
Every complex $K$ has a geometric realization $|K|$ and any two geometric
realizations of a given complex are homeomorphic.
We will assume that every complex has a fixed
geometric realization although, in some cases, we keep the right to
determine the particular choice.


If there is no risk of confusing the reader, we 
write $X$ instead of $|X|$ for a subcomplex $X$ of $K$. 
For example, if we
say that complexes $K$ and $L$ are homeomorphic, we actually mean that $|K|$
and $|L|$ are homeomorphic.

\heading{Subdivisions.}
Let $K,K'$ be simplicial complexes.
We say that $K'$ is a subdivision of $K$
if $|K|=|K'|$ and for each face $\sigma'\in K'$
there is $\sigma\in K$ such that $|\sigma'|\subseteq |\sigma|$. Note that
this definition a priori depends on the choice of the geometric realizations.
However, this is not a problem for us if we fix a realization for every complex
as we mentioned above.



\heading{PL maps and embeddings.}
Let $K,L$ be two simplicial complexes.
A continuous map $|K|\to |L|$
is called {\it PL (piecewise-linear)}
if it is linear on the simplices of
a subdivision $K'$ of $K$.
Then, by~\cite[2.14]{rourke-sanderson72}, there is a subdivision
$L'$ of $L$ such that
$f$ maps any simplex of $K'$ to a simplex of $L'$
and thus induces a simplicial map $V(K')\to V(L')$.

A PL map which is a homeomorphism is called a \emph{PL homeomorphism}.

A \emph{PL $d$-ball} is a simplicial complex PL homeomorphic to the $d$-simplex
$\Delta^d$. A \emph{PL $d$-sphere} is a simplicial complex PL homeomorphic to
the boundary of the $d$-simplex $\partial \Delta^d$. We remark that for $d$
large enough there are known examples of simplicial complexes homeomorphic to
the $d$-ball (resp. $d$-sphere) but which are not a PL $d$-ball (resp. PL
$d$-sphere). 

A \emph{PL embedding} of a simplicial complex $K$ into $\R^d$
is an injective map $|K|\to \R^d$ that is linear on
the faces of $K'$ where $K'$ is some subdivision of $K$. A PL $d$-ball always
PL embeds into $\R^d$ since the $d$-simplex PL embeds into $\R^d$. When we
remove a simplex of maximum dimension from a PL $d$-sphere we obtain a PL
$d$-ball~\cite[Corollary 3.13]{rourke-sanderson72}.\footnote{Note that balls
and spheres in the statement of Corollary 3.13 in~\cite{rourke-sanderson72} are
a priori assumed PL.}

\heading{The Nerve Theorem.}
We need the following version of the Nerve Theorem. The homotopy version
is usually attributed to Borsuk~\cite{borsuk48}. We use the formulation from
Hatcher's book~\cite{hatcher01}.  

\begin{theorem}[{\cite[4G.3]{hatcher01}}]
\label{t:nerve}
If $\mathcal U$ is a collection of open sets in a paracompact space $X$ such
that $\bigcup \mathcal U = X$ and every nonempty intersection of finitely many sets in $\mathcal U$ is contractible, 
then $X$ is homotopy equivalent to the nerve of $\mathcal U$.
\end{theorem}

For further use we recall that any subset of $\R^d$ or $S^d$ is a
paracompact space.

\heading{Homology balls and homology spheres.}
A \emph{homology $d$-sphere} is a (topological) $d$-manifold with the same singular
homology as the $d$-sphere. Similarly, a \emph{homology $d$-ball} is a $d$-manifold with
boundary which has the same singular homology as the $d$-ball.

\heading{Alexander duality and \v{C}ech cohomology.}
As a supplementary tool we also need Alexander duality. Roughly speaking,
Alexander duality relates the cohomology of a ``nice'' closed subset $K$ of $S^d$
with the homology of $S^d \setminus K$. 
If we do not know whether $K$ is
``nice'' (which will be our case), then the ordinary cohomology must be replaced
with \emph{\v{C}ech cohomology}. In order to define \v{C}ech cohomology, we
would need too many preliminaries. Therefore we rather prefer to use it as a
``black box'' while referring to the literature for statements we need.

Here is a version of Alexander duality we need~\cite[Theorem 5.7]{prasolov07}:
\begin{theorem}[Alexander duality]
\label{t:ad}
If $A \subsetneq S^d$ is a closed set, then
$$
\check{\tilde{H}}^k(A) \cong \tilde{H}_{d-k-1}(S^d \setminus A)
$$
for $0 \leq k \leq n-1$. Here $\check{\tilde{H}}^*$ stands for reduced \v{C}ech
cohomology and $\tilde{H}_*$ stands for reduced singular homology.
\end{theorem}

\begin{lemma}[{\cite[exercise 3, p. 254]{eilenberg-steenrod52}}]
\label{l:es}
Let $X \subseteq S^d$. Then the (non-reduced) \v{C}ech cohomology group 
$\check{H}^0(X)$ is isomorphic to the group of continuous functions $X \to \Z$
where $\Z$ is equipped with discrete topology.
\end{lemma}

For clarity, the following lemma summarizes all consequences of Alexander
duality we will need.

\begin{lemma}
\label{l:two_components}
Let $M$ and $N$ be two open subsets of $S^d$, $d \geq 2$. If $M$ is a homology $(d-1)$-sphere and $H_{d-1}(N) = 0$, then
\begin{enumerate}
\item[(a)]
$S^d \setminus M$ contains exactly two components;
\item[(b)]
$S^d \setminus N$ is connected; and
\item[(c)]
$H_{d-1}(M \cup C) = 0$ where $C$ is any of the components of $S^d \setminus
M$.
\end{enumerate}
\end{lemma}

\begin{proof}
We prove the part (a) first.
Let $A := S^d \setminus M$. Then $\check{\tilde{H}}^0(A) \cong \Z$ by Alexander
duality (Theorem~\ref{t:ad}), and therefore $\check{H}^{0}(A) \cong \Z \oplus
\check{\tilde{H}}^0(A) \cong \Z^2$. 

If $A$ contained three or more components then it could be partitioned into three
disjoint clopen (closed and open) sets $A_1$, $A_2$ and $A_3$ (disconnected $A$ can be partitioned
into two clopen sets and then at least one of these sets can be partitioned again).
Functions $A \to \Z$ constant on each of these clopen sets would be continuous.
Therefore $\Z^3$ would be a subgroup of $\check{H}^{0}(A)$ due to
Lemma~\ref{l:es}. This contradicts $\check{H}^{0}(A) \cong \Z^2$.

Similarly, if $A$ were connected, then $\check{H}^{0}(A) \cong \Z$, since every
continuous function $A \to \Z$ would be constant. This contradicts
$\check{H}^{0}(A) \cong \Z^2$ again.

Therefore part (a) is proved. Part (b) is analogous to (a), using
$\check{\tilde{H}}^0(S^d \setminus N) \cong \Z$ which follows from the
Alexander duality (reduced and non-reduced homology groups coincide in
dimension $d-1$, since $d \geq 2$).
It remains to prove (c). 

Let $C'$ be the
second component of $A = S^d \setminus M$. Note that both $C$ and $C'$ are
closed in $S^d$ since $A$ is closed in $S^d$ and the number of components of
$A$ is finite, namely two.
Using Lemma~\ref{l:es} again, we
derive $\check{H}^{0}(C') \cong \Z$, and therefore $\check{\tilde{H}}^{0}(C') = 0$. Part (c) now follows from the Alexander duality.

\end{proof}

\section{The proof method}
\label{sec:planofproof}

In this section we describe our proof method. 
On a  general level, we follow the
approach by Matou\v{s}ek, Tancer and Wagner~\cite{matousek-tancer-wagner11}
showing that it is algorithmically undecidable whether a given
$(d-1)$-dimensional simplicial complex embeds in $\R^d$ (for $d \geq 5$). 
Some details are, however, more difficult to resolve in our case.

Our main ingredient is Novikov's theorem (Theorem~\ref{t:novikov}). 
Using it we construct a sequence of simplicial
complexes $\{C_i\}_{i=1}^{\infty}$ such that each $C_i$
\\
$\bullet$
is either PL homeomorphic to the $d$-ball
\\
$\bullet$
or has nontrivial fundamental group,
\\
and there is no
algorithm deciding which of the two cases holds. 
The main task is to show that
$C_i$ is topologically $d$-representable in the first case
(this is rather straightforward but a bit technical; see
Theorem~\ref{t:emb_rep})
but not in the second
(this is not so obvious
and the reader might be also interested in the used
technique; see Proposition~\ref{p:fund}. It uses
a special
feature of the combinatorial structure of $C_i$;
see {\it collaring} below.)

Now we describe our method in more details. 

\heading{Novikov's theorem.} Novikov proved that it is algorithmically
undecidable whether a given (CW-)complex is homeomorphic to a $d$-sphere if $d
\geq 5$. We need the following variation of his theorem.

\begin{theorem}[\textbf{Novikov}]
\label{t:novikov}
 Let $d\geq 5$ be a fixed integer. There is an effectively
constructible sequence of simplicial complexes
$\Sigma_i$, $i\in \N$, with the following properties:
\begin{enumerate}
\item[\rm (1)] Each $|\Sigma_i|$ is a homology $d$-sphere (in particular a
manifold).
\item[\rm (2)] For each $i$, either $\Sigma_i$ is a \textup{PL} $d$-sphere, or
the fundamental group of $\Sigma_i$ is nontrivial (in particular, $\Sigma_i$ is
not homeomorphic to the $d$-sphere).
\item[\rm (3)] There is no algorithm that decides for every given $\Sigma_i$
which of the two cases holds.
\end{enumerate}
\end{theorem}

A proof of Theorem~\ref{t:novikov} follows from the exposition by
Nabutovsky; see the appendix of~\cite{nabutovsky95}. Indeed Nabutovsky
constructs a sequence of polynomials such that it is algorithmically
undecidable whether their zero set is homeomorphic to a $d$-sphere. These zero
sets are always smooth manifolds, and if they are homeomorphic to a $d$-sphere,
they are in addition diffeomorphic to the standard $d$-sphere. Such smooth
manifolds have a natural PL-structure~\cite{whitehead40} and their
triangulations can be found algorithmically~\cite[Remark
11.19]{basu-pollack-roy06} (see Remark 12.35 if you consult the first edition).
We conclude by remarking that in case of triangulating standard (smooth)
$d$-sphere we obtain a PL-sphere.

Our task is to transform this result into undecidability of recognition
of topologically $d$-representable complexes (for $d \geq 5$).

\heading{Removing a simplex.}
Let $B_i$ be the simplicial complex obtained from 
$\Sigma_i$ by removing a $d$-simplex. 
Each $B_i$ is a homology $d$-ball;
$B_i$ is embeddable into $\R^d$ if and only if $\Sigma_i$
is a PL sphere (which is 
algorithmically unrecognizable).

A straightforward approach 
(motivated by~\cite{matousek-tancer-wagner11}) 
would be to prove the following conjecture.

\begin{conjecture}
\label{conj:emb_equiv_rep}
A simplicial complex $K$ 
is PL embeddable into $\R^d$
if and only if its barycentric subdivision
is topologically $d$-representable.
\end{conjecture}

An affirmative answer to this conjecture implies
Theorem~\ref{t:unrecog} (our main result) if it is used for the barycentric
subdivisions of the sets $B_i$. (For brevity of this part, the definition
of barycentric subdivision is postponed to section~\ref{s:etr}.)


We prove the `only if' part even in a stronger form (Theorem~\ref{t:emb_rep}) 
but we could not
prove the `if' part in such generality,
or even for $K=B_i$.
So we will modify the simplicial complexes
$B_i$ and obtain a new sequence of complexes $C_i$.
Using some new combinatorial features of $C_i$
we are able to prove that they are PL embeddable into $\R^d$ if and only
if they are topologically $d$-representable.


In section~\ref{s:rep_emb} we also prove Conjecture~\ref{conj:emb_equiv_rep} in
case that $\dim K \leq \frac{2d - 3}2$. This range is unfortunately not
sufficient for our main result; however, we still hope that it is an
interesting supplementary result.


\heading{Collaring.} 
Fix $\Sigma_i$ and $B_i$. Let
$U:=\{u_1, \dots, u_{d+1}\}$ be the vertices of the simplex
removed from $\Sigma_i$
and $V=\{v_1, \dots, v_{d+1}\}$ be additional points
(not the vertices of $\Sigma_i$).

We now create a simplicial complex $\Gamma$ with vertices $u_1, \dots,
u_{d+1}$, $v_1, \dots, v_{d+1}$.
The set of simplices of $\Gamma$ is the following:
$$
\{\sigma \subset U \cup V: \sigma \neq U, V \hbox{ and } \{u_j, v_j\}
\not\subseteq
\sigma \hbox{ for any $j \in \{1,\ldots,d+1\}$}\}.
$$

If we did not require $\sigma \neq U,V$ we would obtain a $d$-dimensional
\emph{crosspolytope} (see, e.g.,~\cite[p.~11]{matousek03}). 
Thus $\Gamma$ is isomorphic to a
$d$-dimensional crosspolytope minus two opposite $d$-simplices. In particular,
$\Gamma$ is homeomorphic to $S^{d-1} \times [0,1]$. 
%
%

Now set $C_i := B_i \cup \Gamma$, see Figure~\ref{f:bc}.
Informally, we attached a cylinder (`collar') $\Gamma$ to $B_i$
and obtained $C_i$. Clearly, $C_i$ is homeomorphic to $B_i$.

\begin{figure}
\begin{center}
\includegraphics{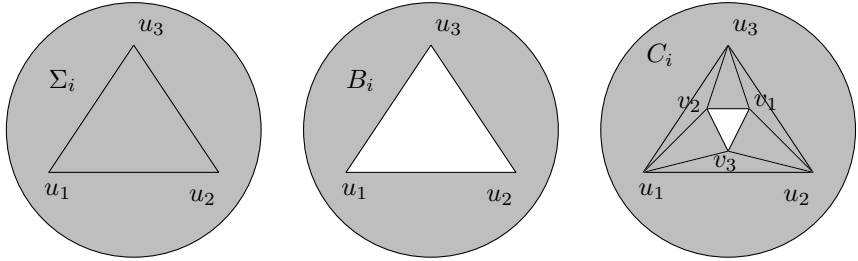}
\caption{The complexes $\Sigma_i$, $B_i$ and $C_i$}
\label{f:bc}
\end{center}
\end{figure}

We are going to show that $C_i$ is topologically
$d$-representable if and only if $\Sigma_i$ is homeomorphic to the $d$-sphere.
We will split this task into two statements proved in separate sections.

\begin{theorem}
\label{t:emb_rep}
Let $K$ be a simplicial complex PL embeddable into $\R^d$. Then
$K$ is topologically $d$-representable by balls (see 
Definition~\ref{d:rep_balls}).
%
\end{theorem}

\begin{proposition}
\label{p:fund}
Let $i$ be such that $\Sigma_i$ has a nontrivial fundamental group. Then $C_i$
is not topologically $d$-representable.
\end{proposition}

Theorem~\ref{t:emb_rep} is proved in section~\ref{s:etr};
Proposition~\ref{p:fund} is proved in section~\ref{s:fund}.
The implication in Theorem~\ref{t:emb_rep} cannot be reverted, since a
simplex of arbitrary high dimension is topologically $d$-representable by
balls; however it is not PL embeddable into $\R^d$ if the dimension of the
simplex exceeds $d$.
Similarly, this example shows that an analogue of Conjecture~\ref{conj:emb_equiv_rep} running as follows:
a simplicial complex is $d$-representable if and only if it $PL$ embeds into $\R^d$, is false.

We conclude this section by summarizing the above mentioned steps into a proof
of Theorem~\ref{t:unrecog}
(and Theorem~\ref{t:unrecog_balls} as well).

\begin{proof}[Proof of Theorem~\ref{t:unrecog} and
Theorem~\ref{t:unrecog_balls}]
Let $\{C_i\}_{i=1}^{\infty}$ be the sequence of simplicial complexes
constructed in this section. 

If $i$ is such that $\Sigma_i$ is not homeomorphic to a $d$-sphere, then $C_i$
is not topologically $d$-representable by Proposition~\ref{p:fund}. (And
therefore $C_i$ is neither topologically $d$-representable by balls.)

If $i$ is such that $\Sigma_i$ is homeomorphic to a $d$-sphere, then $\Sigma_i$
is actually a PL $d$-sphere by Theorem~\ref{t:novikov}. Let $\vartheta :=
\{v_1, \dots, v_{d+1}\}$. Then $C_i \cup \{\vartheta\}$ can be regarded as a
subdivision of $\Sigma_i$, and therefore $C_i \cup \{\vartheta\}$ is a PL
$d$-sphere. Consequently, $C_i$ is a PL $d$-ball~\cite[Corollary
3.13]{rourke-sanderson72}.
So $C_i$ PL embeds into $\R^d$, and hence $C_i$
is topologically $d$-representable by balls by Theorem~\ref{t:emb_rep}
(in particular, it is topologically $d$-representable).

\end{proof}


\section{Embeddable complexes are topologically representable}
\label{s:etr}

In this section we prove Theorem~\ref{t:emb_rep}.

Suppose that $K$ is a simplicial complex and
$f:|K|\to\R^d$ is a PL embedding.
Let $V$ be the set of  vertices of $K$. 
We have to construct 
a topological $d$-representation of $K$, i.e., a family of sets
$\{U_v\}_{v\in V}$, $U_v\subset \R^d$ such that
\begin{description}
\item{(a1)}
the nerve of $\{U_v\}$ is (isomorphic to) $K$; and
\item{(a2)}
the sets $U_v$ and all their intersections
are either homeomorphic to an open $d$-ball or empty.
\end{description}

\heading{Plan of the proof.}
The proof contains two steps.
First, we construct a family $\{X_v\}_{v\in V}$
of certain subcomplexes $X_v\subset K$
such that
\begin{description}
\item{(b1)}
the nerve of $\{|X_v|\}$ is $K$; and
\item{(b2)}
the sets $|X_v|$ and all their intersections
are either (simplicial) cones or empty.
\end{description}
Second, we consider the images $f(|X_v|)\subseteq \R^d$. 
The family $\{f(|X_v|)\}$ has
property (a1), but not property (a2). 
We will introduce $U_v$
as a properly defined open neighborhood of $f(|X_v|)$ in $\R^d$
and show that $\{U_v\}$ is
a good $d$-representation by balls of $K$.

See Figure~\ref{f:embed_repre} while following the
construction.

\heading{Subdivisions and stars.}
Let $L$ be a simplicial complex. We will recall
two notions that we will need further: the {\it barycentric subdivision}
of $L$, and the {\it star} of a vertex $u\in V(L)$.

Formally, the \emph{barycentric subdivision} $\sd K$
of a simplicial complex $K$ is a simplicial complex whose vertices are the faces of
$K$ except the empty face; and the faces of $\sd K$ are 
the chains of faces of $K$
$$
\Lambda = \{\sigma_1, \dots, \sigma_m\} \hbox{ such that }\emptyset \neq
\sigma_1 \subsetneq \sigma_2 \subsetneq \cdots \subsetneq \sigma_m.
$$ 
If there is no risk of confusing reader we simplify the notation by writing
$$
\Lambda = \{\sigma_1 \subsetneq \cdots \subsetneq \sigma_m\}.
$$

In the geometric setting, we can set $|K|=|\sd K|$ 
in such a
way that a vertex of $\sd K$ corresponding to
a simplex $\sigma \in K$ is situated in the
barycentre of $|\sigma|\subset |K|$.

Let $u$ be a vertex of $L$. The (closed)
\emph{star} of $u$ in $L$ is defined as $\st(u, L) := \{\sigma \in L\colon u
\cup \sigma \in L\}$.

\heading{First step. A cover $\{|X_v|\}$ inside $|K|$.}
For each $v\in V$, denote $X_v:=\st(v,\sd K)$.
It is a subcomplex of $\sd K$.
For $S\subseteq V$, denote $X_S:=\bigcap_{v\in S}X_v$.
\begin{figure}[t]
\centering
\includegraphics{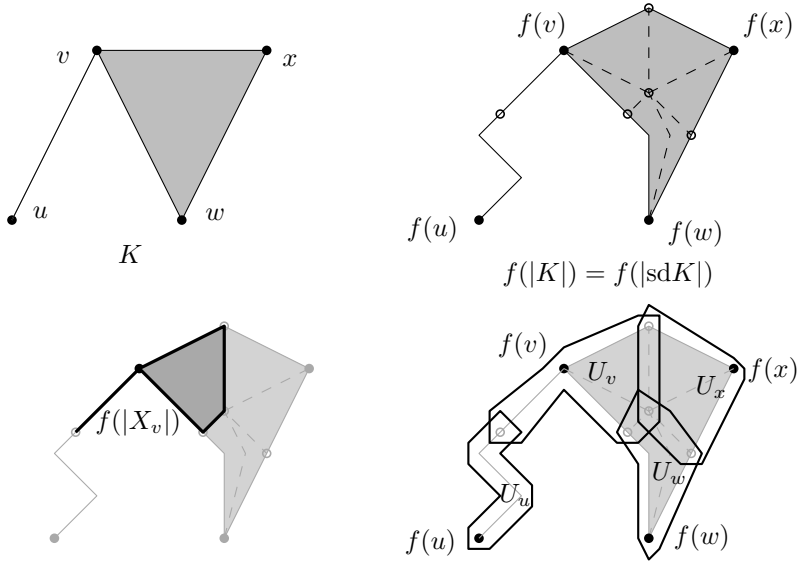}
\caption{A complex $K$ (top left);
a PL embedding $f$ of $|K|$ into $\R^2$ (top right); a~set
$f(|X_v|)$ -- image of the star of $v$ in
the barycentric subdivision of $K$ (bottom left);
and the sets $\{U_v\}$ forming a topological $d$-representation of $K$ (bottom right).}
\label{f:embed_repre}
\end{figure}

The following claim says that the geometric realizations $|X_v|$ of $X_v$ form
a cover with properties (b1) and (b2) announced in the plan of the proof.

\begin{claim}
\label{c:nerve}
For every $S\subseteq V$,
$X_S$ is nonempty if and only if $S\in K$.
If $X_S$ is nonempty, then it is a cone.
\end{claim}

{\it Proof.}
According to the definitions of star and barycentric subdivision
we have:
$$
\st(\{v\}, \sd K) = \left\{ \{ \sigma_1 \subsetneq \cdots \subsetneq
\sigma_m \}\in
\sd K: v \in \sigma_i \hbox{ for every } i \in [m] \right\}.
$$
Therefore
$$
X_S = 
\bigcap\limits_{v \in S} \st(\{v\}, \sd K) = \left\{ \{\sigma_1
\subsetneq \cdots \subsetneq \sigma_m \}\in
\sd K: S \subseteq \sigma_i \hbox{ for every } i \in [m] \right\}.
$$
Hence $X_S$ is nonempty if and only if $S
\in K$. In addition, if $S \in K$, then $\bigcap_{v \in S} \st(\{v\}, \sd K)$
is a cone in $\sd K$ with apex $\{S\}$.
\qed


%
%

\heading{Derived neighborhoods and collapsibility.}
Here we will briefly recall another concept
of PL topology.
Let $L\subset M$ be a simplicial embedding
of a simplicial complex $L$ into a simplicial $d$-manifold $M$.
The {\it derived  neighborhood} of $L$ in $M$
is the subcomplex $N(L)$ of $\sd \sd M$
whose geometric realization $|N(L)|$
is the union of 
all $|\sigma|$ such that
$\sigma\in \sd \sd M$ is a $d$-simplex
and $|\sigma|\cap |L|\neq\emptyset$. See Figure~\ref{f:derived}.

The definition
of {\it collapsible} simplicial complexes
is found in e.g.~\cite[p.39]{rourke-sanderson72}.
We will omit this definition since we use only some properties of
collapsibility described in the following two lemmas.

\begin{lemma}
\label{l:collapse_cone}
\cite[p.40]{rourke-sanderson72}
If a simplicial complex $L$ is a cone
over another simplicial complex,
then $|L|$ is collapsible.
\end{lemma}

\begin{lemma}
\cite[3.27]{rourke-sanderson72}
\label{l:collapse_nbhood}
Let $L\subset M$ be simplicial embedding
of a simplicial complex $L$ into a simplicial $d$-manifold $M$.
If $|L|$ is collapsible, then
$|N(L)|$ is PL homeomorphic to the $d$-ball.
\end{lemma}

We also need Corollary~\ref{c:nbhood_intersect} below which is implied by the
following lemma. The lemma provides a combinatorial description of the derived
neighborhood.

For a simplicial complex $K$ we define a function $\mu\colon \sd K \to
K$ that assigns to a chain $\Lambda \in \sd K$ the minimal element of this
chain. That is, $\mu(\Lambda) = \sigma_1$ if 
$$
\Lambda = \{\sigma_1 \subsetneq \cdots \subsetneq
\sigma_k\} \in \sd K.
$$

\begin{lemma}
\label{l:comb_derived}
Let $L \subset M$ be simplicial embedding of a simplicial complex $L$ into a
simplicial $d$-manifold $M$. Then
$$
N(L) = \{\sigma \in \sd \sd M\colon \mu(\mu(\sigma)) \in L\}.
$$
\end{lemma} 

\begin{proof}
Let $\sigma \in \sd \sd M$. 

We first assume that $\mu(\mu(\sigma)) \in L$ and
we will show that $\sigma \in N(L)$. From the definition of $\mu$ it follows
that $\mu(\mu(\sigma)) \in \mu(\sigma)$; hence
$\{\mu(\mu(\sigma))\}??\subseteq \mu(\sigma)$.\footnote{Note that, purely
formally, if $v$ is a vertex of $M$, then $\{v\}$ is the corresponding vertex
of $\sd M$, and $\{\{v\}\}$ the corresponding vertex of $\sd \sd M$. This
explains the necessity of using iterated parentheses.}
Consequently $h(\sigma) := \sigma
\cup \{\{\mu(\mu(\sigma))\}\}$ is a simplex of $\sd \sd M$. (Note that it might
happen that $h(\sigma) = \sigma$ if $\{\mu(\mu(\sigma))\} = \mu(\sigma)$.) The geometric
realization $|h(\sigma)|$ intersects $|L|$ (in $|\{\mu(\mu(\sigma))\}|$ as a  
vertex of $\sd K$, that is, in the barycentre of the simplex
$|\mu(\mu(\sigma))|$ of $K$). Therefore
$\sigma \in N(L)$.

Before proving the second inclusion, we first realize that if $\vartheta \in M
\setminus L$, then $|\st(\{\{\vartheta\}\}, \sd \sd M)|$ does not
meet $|L|$. This is because $|\st(\{\{\vartheta\}\}, \sd \sd M)| \subseteq
\interior
|\st(\{\vartheta\}, \sd M)|$ (where $\interior$ denotes the interior), and
$\interior |\st(\{\vartheta\}, \sd M)|$ does not meet $|L|$ since $\vartheta
\not \in L$.

Now we assume that $\mu(\mu(\sigma)) \not\in L$. We will show that $\sigma \not
\in N(L)$. That is, we want to show that $|\tau| \cap |L| = \emptyset$ for
every $\tau \in \sd \sd M$ with $\sigma??\subseteq \tau$. See
Figure~\ref{f:derived}. From the definition of
$\mu$ the inclusion $\sigma \subseteq \tau$ implies $\mu(\tau) \subseteq
\mu(\sigma)$. Applying once more, we get $\mu(\mu(\sigma)) \subseteq
\mu(\mu(\tau))$. Therefore $\mu(\mu(\tau)) \not \in L$ since $\mu(\mu(\sigma))
\not \in L$. Similarly as before we have a simplex $h(\tau) := \tau \cup
\{\{\mu(\mu(\tau))\}\}$ containing $\tau$ and therefore $\sigma$ as well. 
However, $|h(\tau)| \cap |L| = \emptyset$ since $h(\tau) \in
\st(\{\{\mu(\mu(\tau))\}\}, \sd \sd M)$, using the observation from the
previous paragraph. 


\end{proof}

\begin{corollary}
\label{c:nbhood_intersect}
Let $L_1,L_2\subset M$ be two simplicial embeddings
of simplicial complexes $L_1,L_2$ into a simplicial $d$-manifold $M$.
Then $N(L_1\cap L_2)=N(L_1)\cap N(L_2)$.
\end{corollary}

\begin{proof}
Let $\sigma \in \sd \sd M$. We have that $\sigma \in N(L_1) \cap N(L_2)$ if and
only if $\mu(\mu(\sigma)) \in L_1$ and $\mu(\mu(\sigma) \in L_2$. This happens
if and only if $\mu(\mu(\sigma)) \in L_1 \cap L_2$, that is, if and only if
$\sigma \in N(L_1 \cap L_2)$.
\end{proof}

\begin{figure}
\begin{center}
\includegraphics{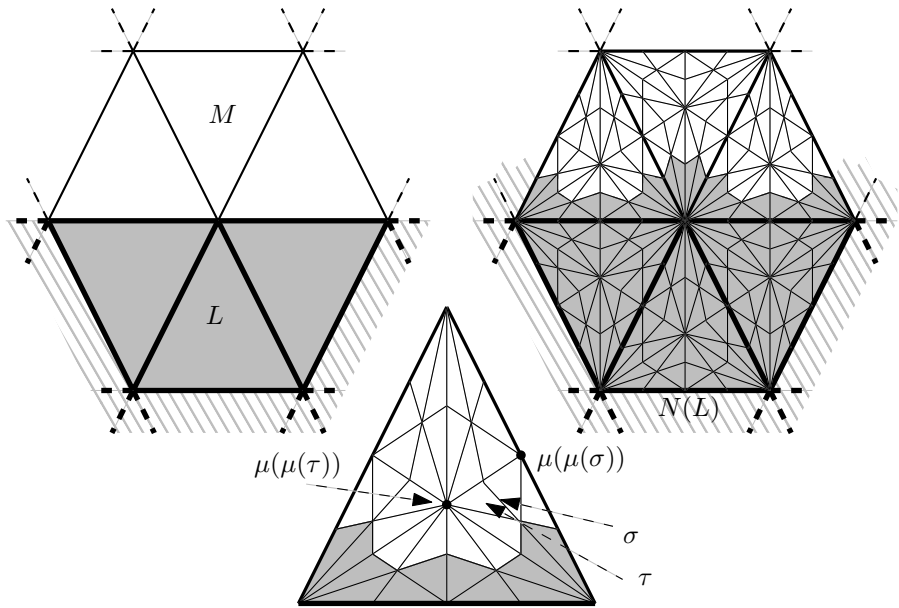}
\caption{Derived neighborhoods: A complex $L$ embedded in a triangulated
manifold $M$ (left). The derived neighborhood $N(L)$ (right). In addition
one of the triangles is enlarged (bottom) with a particular choice of $\sigma$
and $\tau$ such as in the proof of Lemma~\ref{l:comb_derived} (second
inclusion). The notation is simplified; $\mu(\mu(\sigma))$ stands for
$|\{\{\mu(\mu(\sigma))\}\}|$; $\sigma$ stands for $|\sigma|$, etc.
}
\label{f:derived}
\end{center}
\end{figure}

\heading{Second step. A good cover $\{U_v\}$ in $\R^d$.}
By \cite[2.14]{rourke-sanderson72} there is a subdivision of $\sd K$
and a triangulation of $\R^d$ such that
$f$ maps any simplex to simplex in these triangulations. So $f$ induces a
simplicial map between these triangulations as abstract simplicial complexes.
We denote this simplicial map by $f$ again: further $f$ will denote the
simplicial map only.
For $v\in V$, define $U_v:=\mathrm{Int}\, |N(f(X_v))|$.
Here the
derived neighborhood is taken with respect to
the triangulations above.



We conclude the section by proving Theorem~\ref{t:emb_rep}.

\begin{proof}[Proof of Theorem~\ref{t:emb_rep}]
It suffices to prove that the sets $U_v$
obtained above, and all their intersections,
are (either empty or) $d$-balls.
From Claim~\ref{c:nerve}
we know that $X_S = \bigcap_{v\in S}X_v$
is nonempty if and only if $S \in K$. 
If $S \in K$, by Claim~\ref{c:nerve} it is a cone, hence
$X_S$ is collapsible
by Lemma~\ref{l:collapse_cone},
and so is $|\bigcap_{v\in S}f(X_v)|$.
Consequently,
$|N(\bigcap_{v\in S} f(X_v))|$
is a $d$-ball by Lemma~\ref{l:collapse_nbhood}.
Then by Corollary~\ref{c:nbhood_intersect}
$\bigcap_{v\in S}U_v=\mathrm{Int}\, |N(\bigcap_{v\in S} f(X_v))|$
is the interior of a $d$-ball, and thus an open $d$-ball.
\end{proof}

%
%

\section{Nontrivial fundamental group is an obstruction}
\label{s:fund}

\heading{Proof of non-representability of $C_i$ in non-ball
case.} Let us prove Proposition~\ref{p:fund}. 
Let $C_i$ be fixed.
Suppose that there is a good cover in $\R^d$
whose nerve is isomorphic to $C_i$.
Its elements consist of open subsets of $\R^d$
(further called {\it cells}),
each cell corresponding to a vertex of $C_i$.
Let $S^d$, the $d$-sphere, be the 1-point compactification
of $\R^d\subset S^d$.
Any subset in $\R^d$
will be automatically considered as a subset of $S^d$.

Recall that $C_i$ has two special
sets of vertices $U=\{u_i\}_{i=1}^{d+1}$,
$V=\{v_i\}_{i=1}^{d+1}$
belonging to the `collar'
(see Figure~\ref{f:bc}).
Let $\BU_j$ (resp.
$\BV_j$) be the cell 
corresponding to the vertex $u_j$ (resp.
$v_j$) for $j =1,\ldots, d+1$.
Denote $\BU := \BU_1 \cup \cdots \cup \BU_{d+1}$, and $\BV := \BV_1 \cup \cdots \cup
\BV_{d+1}$. Moreover, let $X$ be the set of vertices of $C_i$ minus $U \cup V$
and let $\BX$ be the union of the cells corresponding to
vertices of $X$. 

In our considerations we frequently use the Nerve Theorem
(Theorem~\ref{t:nerve})
without explicitly mentioning it (for instance $\BV$ is homotopy equivalent to
the subcomplex of $C_i$ induced by vertices of $V$, which is homotopy
equivalent to
$S^{d-1}$, etc.). 




By Alexander duality, more precisely by Lemma~\ref{l:two_components}(a),
$S^d \setminus \BV$ has exactly two
components. We know that $\BV$ and $\BX$ are disjoint since there is no edge
connecting a vertex of $X$ with a vertex of $V$ 
(this is the place where we
use the `collar' structure of $C_i$). Thus we can denote by
$V_X$ the component of $S^d \setminus \BV$ containing $\BX$ and by $V_Y$ the
remaining one. See Figure~\ref{f:xuv}.

\begin{claim}
\label{c:vx}
We have $V_X \subseteq \BX \cup \BU$.
\end{claim}

We first derive the result from the claim, and we prove the claim later.

\begin{figure}
\begin{center}
\includegraphics{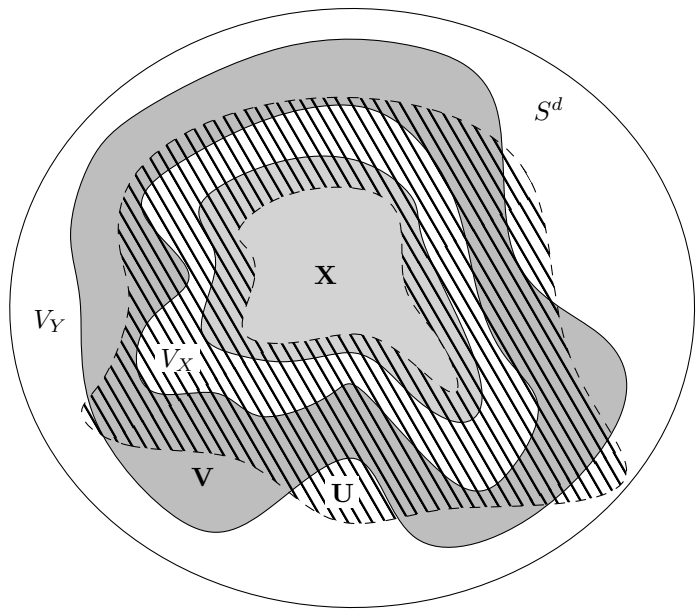}
\caption{The sets $\BU$, $\BV$, $\BX$, $V_X$, and $V_Y$. The set $\BV$ is grey,
$\BU$ is dashed, $\BX$ is light grey, $V_X$ is the component of $S^d \setminus
\BV$ inside $\BV$ and $V_Y$ is the other component.}
\label{f:xuv}
\end{center}
\end{figure}



Let us set $L := \BU 	\cup \BX$ and $M := \BU \cup \BV \cup V_Y$.
We have $L \cup M = S^d$ by Claim~\ref{c:vx}. We also have $L \cap M = \BU$,
since $\BX$ and $V_Y$ are disjoint by the definition of $V_Y$, and we have
already observed that $\BX$
and $\BV$ are disjoint. Both $L \cup M$ and $L \cap M$ have trivial fundamental
group, thus $L$ also has trivial fundamental group by Seifert--van~Kampen theorem.

On the other hand $L$ must have a nontrivial fundamental group, since
it is homotopy equivalent to $B_i$ by the Nerve Theorem. We obtain a 
contradiction as soon as we prove Claim~\ref{c:vx}.


In order to prove Claim~\ref{c:vx} we need two other auxiliary claims.
The first one does not seem new,
but we could not find a reference for it.

\begin{claim}
\label{c:funct}
Let $A\subset B$ be two simplicial complexes
which are represented by 
good covers, the cover representing $A$ 
being a subcover of the cover representing $B$.
Let $\mathbf{A}$ (resp.~$\mathbf{B}$)
be the union of all sets in the representation of $A$
(resp.~$B$). Then the following diagram is commutative,
in which the horizontal maps are inclusion-induced and the vertical maps
are the isomorphisms induced by the homotopy equivalence from the Nerve
Theorem.
$$
\begin{CD}
H_k(A)@>>>H_k(B)\\
@VV\cong V @VV\cong V\\
H_k(\mathbf{A})@>>>H_k(\mathbf{B})
\end{CD}
$$
\end{claim}

\begin{claim}
\label{c:notcover}
We have $V_Y \not \subseteq \BU$.
\end{claim} 

\begin{proof}[Proof of Claim~\ref{c:funct}]
For each cover
(let it be the cover corresponding to $A$),
there can be constructed~\cite[4G]{hatcher01}
a space $\Delta(A)$
together with projections
$\pr_A:\Delta(A)\to A$
and
$\pr_{\mathbf{A}}:\Delta(A)\to \mathbf{A}$
which are homotopy equivalences.
(This is how the Nerve Theorem is generally proved.)
Apply the same construction to $B$.
It is easy to see from the definitions~\cite[4G]{hatcher01} that
we get $\Delta(A)\subset \Delta(B)$,
$\pr_A=\pr_B|_{\Delta(A)}$ and
$\pr_{\mathbf{A}}=\pr_{\mathbf{B}}|_{\Delta(A)}$.
We thus obtain the following commutative diagram:
$$
\begin{CD}
A@>\subset>>B\\
@A\pr_{A} A\sim A @A\pr_B A\sim A\\
\Delta(A)@>>>\Delta(B)\\
@V\pr_{\mathbf{A}}V\sim V @V\pr_{\mathbf{B}}V\sim V
\\
\mathbf{A}@>\subset>>\mathbf{B}
\end{CD}
$$
The vertical maps are homotopy equivalences~\cite[4G]{hatcher01}.
Passing to homology, we obtain the commutative diagram
$$
\begin{CD}
H_k(A)@>>>H_k(B)\\
@A(\pr_{A})_* A\cong A @A(\pr_B)_* A\cong A\\
H_k(\Delta(A))@>>>H_k(\Delta(B))\\
@V(\pr_{\mathbf{A}})_*V\cong V @V(\pr_{\mathbf{B}})_*V\cong V
\\
H_k(\mathbf{A})@>>>H_k(\mathbf{B})
\end{CD}
$$
where the vertical maps are isomorphisms.
\end{proof}

\begin{proof}[Proof of Claim~\ref{c:notcover}]
Recall the subcomplex $\Gamma\subset C_i$
from the collaring procedure.
Let $\Gamma[V]$ be the subcomplex
of $\Gamma$ generated by the set of vertices $V$.
In this proof,
we will abuse notation and write $\Gamma$, $\Gamma[V]$
instead of $|\Gamma|$, $|\Gamma[V]|$.
We apply Claim~\ref{c:funct}
taking $A=\Gamma[V]$, 
$B=\Gamma$ and 
$k=d-1$.
We obtain the following commutative diagram.
$$
\begin{CD}
H_{d-1}(\Gamma[V])@>f>>H_{d-1}(\Gamma)\\
@VV\cong V @VV\cong V\\
H_{d-1}(\mathbf{V})@>g>>H_{d-1}(\mathbf{U}\cup \mathbf{V})
\end{CD}
$$

As follows from the definition of $\Gamma$,
the map $f$ is the isomorphism $\Z\stackrel{\cong}\to \Z$.
On the other hand, if $V_Y\subset \mathbf{U}$, then
$g$ is the zero map.
Indeed, under this assumption $g$ is the composition of the inclusion-induced
maps
$$H_{d-1}(\mathbf{V})\to H_{d-1}(\mathbf{V}\cup V_Y)\to H_{d-1}(\mathbf{U}\cup
\mathbf{V})$$
with $H_{d-1}(\mathbf{V}\cup V_Y)=0$ due to
Lemma~\ref{l:two_components}(c),
so $g=0$.
This contradicts to the commutativity of the diagram.
\end{proof}

\begin{proof}[Proof of Claim~\ref{c:vx}]
The set $\BV \cup \BU \cup \BX$ has trivial $(d-1)$st homology, since it
is homotopy equivalent to $C_i$ which is a homology ball. Hence $S^d \setminus
(\BV \cup \BU \cup \BX)$ is connected due to Lemma~\ref{l:two_components}(b).
Thus $\BV \cup \BU \cup \BX$ has to contain (exactly) one
of the components of $S^d \setminus \BV$. Hence $\BU \cup \BX$ has to contain
$V_X$ or $V_Y$. In addition $\BX$ is disjoint with $V_Y$ by the definition of
$V_X$ and $\BU$ does not cover $V_Y$ by Claim~\ref{c:notcover}. The only
remaining option is that $\BU \cup \BX$ covers $V_X$.
\end{proof}







\section{Topological $d$-representability in the metastable range}
\label{s:rep_emb}
In this section we prove Conjecture~\ref{conj:emb_equiv_rep} if $\dim K \leq
\frac{2d - 3}3$. More precisely, we prove the following result since the
converse implication is already covered by Theorem~\ref{t:emb_rep}:

\begin{theorem}
\label{t:rep_emb}
Assume that $K$ is a $k$-dimensional simplicial complex with 
$k \leq \frac{2d - 3}3$. If $\sd K$, or any subdivision of $\sd K$, is topologically $d$-representable, then
$K$ PL embeds into $\R^d$.
\end{theorem}

The assumption $k \leq \frac{2d-3}3$ is known as that the pair
$(k,d)$ belongs to the \emph{metastable range} of a theorem of Haefliger and
Weber. 
The contents of this section can be regarded as an extension of methods
used in~\cite{tancer11dimension}.

We need some preliminaries.

\heading{Haefliger-Weber Theorem}. 
Let $X$ be a compact topological space. The \emph{deleted product} of a topological
space $X$ is the Cartesian product of $X$ with itself minus the diagonal:
$$
\tilde{X} := X \times X \setminus \{(x,x)\colon x \in X \}. 
$$
There is a natural $\Z_2$-action on $\tilde{X}$ given by swapping coordinates:
$(x,y) \rightarrow (y,x)$. In sequel we assume that $\tilde{X}$ is equipped
with this $\Z_2$-action. By $S^{d-1}_-$ we also denote $(d-1)$-dimensional
sphere equipped with the antipodal action $x \rightarrow -x$.

Let us assume that there exists an embedding $f\colon X \rightarrow \R^d$. The
\emph{Gauss map} $\tilde{f}\colon \tilde{X} \rightarrow S^{d-1}_-$ is the
$\Z_2$-equivariant map given by formula
$$
\tilde{f}(x,y) = \frac{f(x) - f(y)}{\|f(x) - f(y)\|}.
$$
Therefore we know that the existence of embedding $X$ into $\R^d$ 
implies the existence of $\Z_2$-equivariant map from $\tilde{X}$ to $S^{d-1}_-$.

The celebrated Haefliger-Weber Theorem (\cite{haefliger63,weber67}; see
also~\cite{skopenkov08}) states that for polyhedra in the
metastable range the existence of an embedding and the existence of the
equivariant map are equivalent:

\begin{theorem}[Haefliger-Weber]
Let $X$ be a geometric realization of a $k$-dimensional simplicial complex. Let
us also assume that $k \leq \frac{2d-3}3$. If there is a $\Z_2$-equivariant
map $\tilde{X} \rightarrow S^{d-1}_-$, then $X$ is PL embeddable into $\R^d$. 
\end{theorem}

\heading{Weakly injective maps and embeddings.}
Let $K$ be a simplicial complex. We say that a map $f \colon |K| \rightarrow
\R^d$ is \emph{weakly injective} (with respect to $K$) if for every two disjoint simplices $\gamma,
\delta \in K$ their images $f(|\gamma|)$ and $f(|\delta|)$ are disjoint as well.

\begin{remark}
Note that every injective map is weakly injective, but the converse is not
true. In a weakly injective map the images of two faces sharing a vertex might intersect also in other
points. Also the image of a single face might be self-intersecting or even
degenerate.
\end{remark}

For our purposes we need the following corollary of the Haefliger-Weber
Theorem:

\begin{corollary}
\label{c:wi_emb}
Let $K$ be a $k$-dimensional simplicial complex and $d$ be such that $k \leq
\frac{2d-3}3$. Then the existence of a weakly injective map $f \colon |K|
\rightarrow \R^d$ implies the existence of a PL embedding $|K|
\rightarrow \R^d$.
\end{corollary}

\begin{proof}
A simplicial deleted product of $|K|$ is a topological space consisting of
products of pairs of disjoint simplices in $|K|$:
$$
|\tilde{K}|_{s} := \{|\sigma| \times |\tau| \colon \sigma, \tau \in K; \sigma
\cap \tau = \emptyset \}.
$$

The existence of $f$ implies that there is a $\Z_2$-equivariant map
$\tilde{f}_s \colon |\tilde{K}|_{s} \rightarrow S^{d-1}_-$ similarly as
the
existence of an embedding implies the existence of the Gauss map.

It is known that the simplicial deleted product $|\tilde{K}|_{s}$ is
equivariantly homotopic to the deleted product $|\tilde{K}|$;
see~\cite[remark below Example 3.3]{melikhov09}  and the references therein.
Thus there is also a $\Z_2$-equivariant map $|\tilde{K}| \rightarrow
S^{d-1}_-$. Therefore $|K|$ PL embeds into $\R^d$ by the Haefliger-Weber
theorem.
\end{proof}

\heading{Towards a weakly injective map from topological representation.}

Let $\{U_i\}$ be a good cover in $\R^d$ and $L$ be the nerve of this good
cover. In the following lemma we will establish the existence of a certain
auxiliary map $g\colon|L| \rightarrow \R^d$. In order to state the properties
of $g$, we need few preliminaries.

We say that two faces $\alpha, \beta$ in $L$ are
\emph{remote} if there is no edge $\{a,b\} \in L$ such that $a \in \alpha$ and
$b \in \beta$.
We also emphasize here a certain notational issue. We recall that the
vertices of $L$ are the sets $U_i$. Therefore it make sense to consider the
unions of faces in $L$. For example, if $\alpha := \{U_1,U_2\} \in L$, then
$$
\bigcup \alpha = \bigcup_{U_i \in \alpha} U_i = U_1 \cup U_2.
$$


\begin{figure}
\begin{center}
\includegraphics{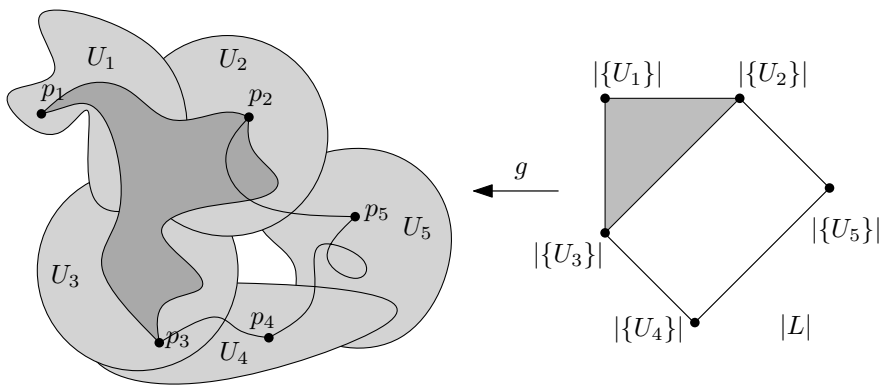}
\end{center}
\caption{Map $g$ from $|\sd L|$ into $\R^d$. The notation is slightly
simplified. For example $p_1$ stands for $g(|\{U_1\}|)$, etc.}
\label{f:map_g}
\end{figure}


\begin{lemma}

\label{l:map_g}

Let $\{U_i\}$ be a good cover in $\R^d$ and $L$ be its nerve. Then there is a map $g \colon |L| \rightarrow \R^d$ such that
\begin{enumerate}
\item[(i)] $g(|\sigma|) \subseteq \bigcup \sigma$ for each $\sigma \in
L$;
and
\item[(ii)] $g(|\alpha|) \cap g(|\beta|) = \emptyset$ for any two
remote $\alpha, \beta \in L$.
\end{enumerate}
\end{lemma}

\begin{proof}
See Figure~\ref{f:map_g} while following the proof.

First we specify $g$ on vertices of $L$. Then we extend it inductively 
to higher dimensional simplices of $L$.

A vertex of $L$ is one of the sets $U_i$. We set $g(|\{U_i\}|)$ to be an arbitrary
point inside $U_i$. 
%
Note that (i) is satisfied for vertices of $L$.

Now we inductively assume that $g$ is defined on all simplices of $L$
of dimension at most $k - 1$. Our task is to extend $g$ to all simplices of
$L$ of dimension $k$. We also assume that condition (i) is valid for
all $\sigma' \in L$ of dimension at most $k-1$.

Let $\sigma$ be a $k$-simplex of $L$. From condition (i) we know that the
$g$-images of all
proper subfaces of $\sigma$ belong to $\bigcup \sigma$, so $g(|\partial
\sigma|)\subset \bigcup \sigma$. But $\partial \sigma$ is homeomorphic to
the $(k-1)$-sphere and $\bigcup \sigma$ is contractible due to the
Nerve Theorem. So we can extend $g$ defined on $|\partial \sigma|$
 to a PL map $g\colon |\sigma|\to \bigcup \sigma$. To complete the inductive step, we
extend $g$ in this way to every $k$-simplex $|\sigma|$. Note that condition (i)
is satisfied by construction.

We have defined $g$ so that it satisfies condition (i). It remains to show that
it satisfies (ii) as well. Let $\alpha$ and $\beta$ be remote simplices of $L$.
By condition (i), $g(|\alpha|) \subset \bigcup \alpha$ and
$g(|\beta|)\subset \bigcup \beta$. If the two right-hand unions had any
intersection, this would mean there exist $k,l$ such that $U_k\in\alpha$,
$U_l\in\beta$ and $U_k\cap U_l\neq \emptyset$. But this means
$\{U_k,U_l\}\in L$, so $\alpha$ and $\beta$ are not remote.
\end{proof}

\begin{proof}[Proof of Theorem~\ref{t:rep_emb}]
Let us assume that $L$ is some subdivision of $\sd K$ that is topologically $d$-representable. Let $\G$ be a topological $d$-representation of $L$. For simplicity of
notation, we assume that $L$ is the nerve of $\G$. Let $g\colon  |L|
\rightarrow \R^d$ be the map
from Lemma~\ref{l:map_g}. Our task is to show that $g$ is weakly injective
with respect to $K$.

Let $\gamma$ and $\delta$ be disjoint simplices of $K$. Let $\alpha$ be a
simplex of $L$ with $|\alpha| \subseteq |\gamma|$ and $\beta$ be a
simplex of $L$ with $|\beta| \subseteq |\delta|$. Then $\alpha$
and $\beta$ are remote in $L$ since in particular $|\alpha| \subseteq
|\gamma'|$ where $\gamma'$ is some simplex of $\sd \gamma$ and similarly with
$\beta$ and $\sd \delta$. Thus $g(|\sd \alpha|) \cap g(|\sd \beta|) =
\emptyset$ by Lemma~\ref{l:map_g}. Consequently $g(|\gamma|) \cap g(|\delta|) =
\emptyset$ for any choice of disjoint $\gamma$ and $\delta$. Therefore $g$ is
weakly injective.

We conclude by stating that Corollary~\ref{c:wi_emb} implies that $K$ PL embeds
into $\R^d$.
\end{proof}

\begin{remark}
Note that in the proof of Theorem~\ref{t:rep_emb} we only need that $L$ is a
``sufficiently fine'' subdivision in the following sense: if $\gamma$ and
$\delta$ are disjoint simplices of $K$ and if $\alpha$ and $\beta$ are
simplices of $L$ satisfying $|\alpha| \subseteq |\gamma|$, and $|\beta| \subseteq
|\delta|$, then $\alpha$ and $\beta$ are remote in $L$. Therefore,
Theorem~\ref{t:rep_emb} can be furthermore extended to such subdivisions.
\end{remark}






\section{Further questions}
We have proved that for $d \geq 5$ 
it is algorithmically undecidable whether a given simplicial complex is
topologically $d$-representable. In our proof we have
used simplicial complexes of dimension $d$. It is natural to ask whether the
recognition of topologically $d$-representable simplicial complexes becomes
algorithmic if we pose some additional restrictions on these complexes.

On the positive side, there is even a polynomial algorithm deciding whether a
given $d/2$-dimensional simplicial complex embeds into $\R^d$ (for $d \geq 6$
even, or $d= 2$).
This is true because Van Kampen's obstruction is a complete obstruction for
embeddability in this range and it is computable in a polynomial time;
see~\cite{matousek-tancer-wagner11} for more details. Therefore by 
Theorems~\ref{t:emb_rep} and \ref{t:rep_emb} we have the following corollary:\\

\begin{corollary}
\label{c:baryc_decidable}
Let $K$ be a simplicial complex of dimension $\frac{d}2$ with $d \geq 6$
even.
Then there is a polynomial time algorithm deciding whether $\sd K$ 
is topologically $d$-representable.\footnote{Theorem \ref{t:rep_emb} can be extended to the case $k=1$,
$d=2$ if we use Hanani-Tutte theorem instead of Haefliger-Weber theorem.
However, this is only a marginal improvement, therefore we do not include it
here separately. Then we could include the case $d=2$ in the corollary as well.}
\end{corollary}

If $K$ is $k$-dimensional instead of specifically $\frac d2$-dimensional, it is
in general not known whether there is an algorithm deciding whether $K$
PL embeds into $\R^d$. However, based on work of \v{C}adek et
al~\cite{cadek-krcal-matousek-sergeraert-vokrinek-wagner11soda}, it is plausible to believe that this
embeddability question is decidable for all pairs $(k,d)$ in the metastable
range. If this is true, then Corollary~\ref{c:baryc_decidable} can be extended
(maybe without the polynomial time estimate) to the whole metastable range.

It should be emphasized that it is quite restrictive to look for an algorithm
with restrictions on the triangulation of the complex. Therefore it is
natural to ask what happens if we pose only dimensional restrictions:

\begin{question}
For which pairs of integers $k$ and $d$ is there an algorithm which recognizes
whether a given simplicial complex of dimension at most $k$ 
is topologically $d$-representable?
\end{question}

\begin{remark}
A simplicial complex
$K$ is topologically $d$-representable if and only if
 the disjoint union of $K$ and a
simplex of arbitrary high dimension is topologically $d$-representable.
Therefore `at most $k$' can be replaced with `exactly $k$' without changing the
outcome.
\end{remark}

Our main result says that the answer is no if $5 \leq d \leq k$.

If $d \geq 2k + 1$, then every simplicial complex of dimension at most $k$ is
topologically $d$-representable. This follows, for example, from
Theorem~\ref{t:emb_rep} and the fact that every $k$-dimensional simplicial
complex is even linearly embeddable into $\R^{2k+1}$.

If $d=1$, then it is not so hard to see that the answer is yes no matter what
is $k$, because topologically $1$-representable complexes are \emph{clique
complexes over interval graphs}.

For other pairs $(k,d)$ we do not know the answer. It would be especially
interesting if there was an algorithm in the whole metastable range.

\section*{Acknowledgment}
We would like to thank Thomas Goodwillie,
Roman Karasev,
Ji\v{r}{\'\i}  Matou\v{s}ek, Sergey Melikhov, 
Paul Siegel, Arkadiy Skopenkov and Uli Wagner for 
fruitful discussions and/or kind answers to our
questions.

\bibliographystyle{alpha}
\bibliography{/home/martin/clanky/bib/general}

\appendix
\section{Proof of Theorem~\ref{t:unrecog_acyclic}}
\label{app_acyclic}
Recall that our main result, Theorem~\ref{t:unrecog}, comes with two
supplementary variations, Theorems~\ref{t:unrecog_balls}
and~\ref{t:unrecog_acyclic}. In this section prove
Theorem~\ref{t:unrecog_acyclic}, the other two theorems being already proved.
We heavily rely on the notation from the previous parts 
of the paper, especially from Section~\ref{s:fund}.

To prove Theorem~\ref{t:unrecog_acyclic}, it clearly suffices to prove the following generalization of Proposition~\ref{p:fund}. (See the proof of Theorems~\ref{t:unrecog} and~\ref{t:unrecog_balls} in Section~\ref{sec:planofproof}.) 

\begin{proposition}
\label{p:fund_acyclic}
Let $i$ be such that $\Sigma_i$ has a nontrivial fundamental group. Then $C_i$
is not $d$-representable by an acyclic cover.
\end{proposition}

The proof of Proposition~\ref{p:fund} is given in Section~\ref{s:fund} above. We prove Proposition~\ref{p:fund_acyclic} by changing the necessary places from the proof of Proposition~\ref{p:fund}. 

\begin{proof}
Suppose $\pi_1(\Sigma_i)\neq 0$, but there is an acyclic cover $\{U_i\}$ representing $C_i$. We need to come to a contradiction.
The Nerve theorem used several times in the proof of Proposition~\ref{p:fund}
is inapplicable in the current situation. Our plan is
to trace every appearance of the Nerve theorem in the proof of Proposition~\ref{p:fund}, and reprove the conclusions derived from the Nerve theorem using a different argument. If all such conclusions are proved by arguments valid for the acyclic cover $U_i$, Proposition~\ref{p:fund_acyclic} is proved.

The proof of Proposition~\ref{p:fund} first uses
the Nerve theorem to show that $S^d\setminus\BV$ has exactly two components.
Here we can use Leray's homological version of the Nerve
theorem~\cite{leray45}; see, e.g., also Theorem~2.1 of~\cite{meshulam01}.

\begin{theorem}
\label{t:homology_acyclic}
Let $\{U_i\}$ be an acyclic cover in $\R^d$. Then the singular $\Z$-homology groups of $\bigcup_i U_i$ are isomorphic to those of the nerve $N(\{U_i\})$ of the cover.
\end{theorem}
This allows to apply Lemma~\ref{l:two_components}(a) as in the original proof to conclude that $S^d\setminus \BV$ has two components.

The next places where the Nerve theorem is used are Claims~\ref{c:vx},~\ref{c:funct} and~\ref{c:notcover}.
Below we prove the following analogue of Claim~\ref{c:funct} (which is again most probably known, but we could not find a reference for it).

\begin{claim}
\label{c:funct_acyclic}
Let $A\subset B$ be two simplicial complexes
which are represented by 
acyclic covers, the cover representing $A$ 
being a subcover of the cover representing $B$.
Let $\mathbf{A}$ (resp.~$\mathbf{B}$)
be the union of all sets in the representation of $A$
(resp.~$B$). Then the following diagram is commutative,
in which the horizontal maps are inclusion-induced and the vertical maps
are isomorphisms.
$$
\begin{CD}
H_k(A)@>>>H_k(B)\\
@VV\cong V @VV\cong V\\
H_k(\mathbf{A})@>>>H_k(\mathbf{B})
\end{CD}
$$
\end{claim}

The original proof of Claims~\ref{c:notcover} and~\ref{c:vx} become valid for an acyclic cover $\{U_i\}$ if we use Claim~\ref{c:funct_acyclic} instead of Claim~\ref{c:funct} and Theorem~\ref{t:homology_acyclic} instead of the Nerve theorem.

In the proof of Proposition~\ref{p:fund}, the passage after the statement of Claim~\ref{c:notcover} is the last place where the Nerve theorem is used.
There we introduced two sets $L=\BU\cup\BX$ and $M$ such that $L\cup M=S^d$
and $L\cap M=\BU$.
We know that $\pi_1(U)=0$ and $\pi_1(U\cup X)\neq 0$.
If $\{U_i\}$ is an acyclic cover, we first prove the claim
below.

\begin{claim}
\label{c:pi1_notsurjective}
The
inclusion-induced map $i: \pi_1(L\cap M)\to\pi_1(L)$ is not surjective.
\end{claim}

Now we see from Seifert-van Kampen's theorem that the group
$\pi_1(L\cup M)$ has a quotient isomorphic to $\pi_1(L)/\mathrm{Im}\,
i$ which is non-zero by Claim~\ref{c:pi1_notsurjective}. On the other hand,
$\pi_1(L\cup M)=\pi_1(S^d)=0$, a contradiction.
Proposition~\ref{p:fund_acyclic} is proved modulo Claims~\ref{c:funct_acyclic} and~\ref{c:pi1_notsurjective}.
\end{proof}

To prove the remaining claims, let us recall an explicit construction
of $\Delta$-sets that appeared previously in the proof of Claim~\ref{c:funct}.
\begin{definition}
Let $\{U_i\}$ be an arbitrary (finite) open cover, i.e.,~a collection of open
sets in $\R^d$. Let $N = N(\{U_i\})$ be the nerve of this cover. For $\sigma
\in N$ we let $U_\sigma$ to denote the intersection of all $U_i$ corresponding
to the vertices of $\sigma$. For further use, we also set $U_{\emptyset} =
\bigcup_i U_i$.
We define $\Delta(\{U_i\})$ as a subset of $|N| \times
\bigcup_i U_i$ given by $\bigcup_{\sigma \in N} (|\sigma| \times U_{\sigma})$. 

There are two natural projections $\pr_{N}\colon\Delta(\{U_i\})\to |N|$
coming as the projection to the first factor and 
$\pr_{\bigcup_i U_i}\colon\Delta(\{U_i\})\to \bigcup_iU_i$
coming from the second factor.
\end{definition}



These projections yield homotopy equivalences as described in the following
lemma (we have already used these homotopy equivalences in the proof of
Claim~\ref{c:funct}).

\begin{lemma}
\label{l:delta_homotopy}

\begin{itemize}
\item[(a)]
For any cover $\{U_i\}$, the map $\pr_{\bigcup_iU_i}$ is a homotopy equivalence
\cite[Proposition 4G.2]{hatcher01},
\cite[proof of Theorem 3.21, Step 1]{prasolov06}.

\item[(b)]
For a good cover $\{U_i\}$, the map $\pr_N$ is a homotopy equivalence
\cite[Colollary 4G.3]{hatcher01},
\cite[proof of Theorem 3.21, Step 2]{prasolov06}.
\end{itemize}
\end{lemma}

We also need a supplementary construction turning an acyclic cover into a good
cover (in higher dimensional space) while keeping the nerve. The following
lemma summarize an induction step.

\begin{lemma}
\label{l:ind_acyclic}
Let $\{U_i\}_{i=1}^n$ be an acyclic cover in $\R^d$, $F$ be a filter on
$N = N(\{U_i\})$ (that is, $F \subseteq N$ and if $\sigma' \supseteq
\sigma \in F$, then $\sigma' \in F$) and
$\vartheta$ be a nonempty inclusionwise maximal element of $N\setminus F$.
Let us assume that $U_\sigma$ is contractible for every $\sigma \in F$. Then
there is an open cover $\{\hat U_i\}_{i=1}^n$ in $\R^{d+1}$ satisfying the
following properties.

\begin{enumerate}
\item $U_i \subseteq \hat U_i$ and this inclusion induces an isomorphism
between nerves $N$ and $\hat N := N(\{\hat U_i\})$.
\item $\hat U_{\hat \sigma}$ is contractible for every $\sigma \in \hat F := F
\cup \{\vartheta\}$ where $\hat \sigma$ is an image of $\sigma$ via the
isomorphism from property 1.
\item The inclusion $U_\sigma \subseteq \hat
U_{\hat \sigma}$ induces an isomorphism in all homology groups; in particular
the cover $\{\hat U_i\}$ is acyclic. (Here we also allow $\sigma =
\emptyset$, so the inclusion $\bigcup_i U_i \subseteq \bigcup_i \hat U_i$
induces an isomorphism on all homology groups as well.)
\end{enumerate}

\end{lemma}

\begin{proof}
We set
\begin{itemize}
\item $\hat U_i := U_i \times (0,1)$ if $U_i \notin \vartheta$;
\item $\hat U_i := U_i \times (0,1) \cup \cone((*,2), U_{\vartheta} \times
\{1\})
\cup B((*,2), \frac12)$ if $U_i \in \vartheta$. Here $*$ is an arbitrary (fixed) point of $\R^d$,
$\cone(a,X)$ denotes the cone with apex $a$ and basis $X$, and $B(c,r)$ denotes
the open ball with center $c$ and radius $r$. See Figure~\ref{f:ac_gc}. 
\end{itemize} 

\begin{figure}
\begin{center}
\includegraphics{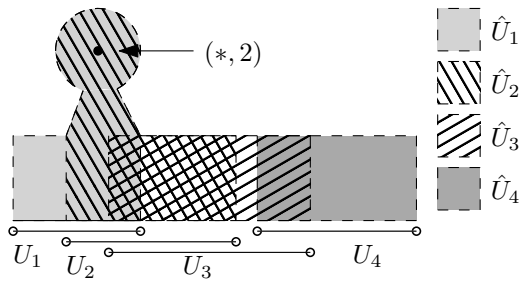}
\caption{Construction of $\{\hat U_i\}$ from $\{
U_i\}$ assuming $\vartheta = \{U_1,U_2\}$.}
\label{f:ac_gc}
\end{center}
\end{figure}

Now obviously $\hat U_i$ are open sets and $U_i \subseteq \hat U_i$ if we identify $U_i$ with $U_i \times
\{\frac12\}$. We consecutively check the properties. 

Given $\sigma = \{U_{i_1}, \dots, U_{i_k}\}$ we let $\hat \sigma := \{\hat
U_{i_1}, \dots, \hat U_{i_k}\}$. If $\sigma \in N$ then $\hat \sigma \in \hat N$ since $U_i \subseteq \hat U_i$. On the other hand, if $\hat \sigma = \{\hat U_{i_1}, \dots, \hat
U_{i_k}\}$ belongs to $\hat N$ then there is a witness $x \in \bigcap_k \hat
U_{i_k}$. Assuming $x = (x',x_{d+1}) \in \R^d \times \R$ we have either 
$x_{d+1} \in (0,1)$ which obviously implies $\sigma \in N$ or $x_{d+1} \in
[1,2.5)$ which implies $\hat \sigma \subseteq \hat \vartheta$; therefore $\sigma
\in N$.  

Next let us assume that $\sigma \in N$. If $\sigma \in F$, then $\hat U_{\hat
\sigma} = U_\sigma \times (0,1)$, therefore $\hat U_{\hat \sigma}$ is
contractible. We also have $\hat U_{\hat \vartheta} = 
U_{\vartheta} \times (0,1) \cup \cone((*,2) \cup B((*,2),\frac12)$ which is
contractible.

Finally, we show that the inclusion $U_\sigma \times (0,1) \subseteq \hat
U_{\hat \sigma}$
induces an isomorphism on homology groups which is sufficient since the
inclusion $U_i \subseteq U_i \times (0,1)$ obviously induces an isomorphism
(recalling identification of $U_i$ and $U_i \times \{\frac12\}$). 
Let us also assume that $\sigma \subseteq \vartheta$, otherwise $U_\sigma
\times (0,1) = \hat U_{\hat \sigma}$.
From the
exact sequence of the pair, it is sufficient to show that the homology of the
pair $(\hat U_{\hat \sigma}, U_{\sigma} \times (0,1))$ vanishes. We have
\begin{multline*}
H_k(\hat U_{\hat \sigma},U_{\sigma} \times (0,1))
\cong 
\tilde H_k(\hat U_{\hat \sigma}/(U_{\sigma} \times (0,1)))
\cong
\tilde H_k(\Sigma (U_{\vartheta}))
\cong 
\tilde H_{k-1}(U_{\vartheta})=0,
\end{multline*}
where $\Sigma$ denotes the suspension. The last equality holds because
$\{U_i\}$ is an acyclic cover.

\end{proof}

\begin{proof}[Proof of Claim~\ref{c:funct_acyclic}]
With the construction of $\{\hat U_i\}$ from $\{U_i\}$ at hand, we are ready to prove Claim~\ref{c:funct_acyclic}.

Starting with an acyclic cover $\{\hat U_i\}$ we first set $F = \emptyset$.
Then we repeatedly apply Lemma~\ref{l:ind_acyclic} adding to the filter an 
inclusionwise maximal $\vartheta$ which is not in the filter yet. After $|N| - 1$
steps we obtain an acyclic cover $\{\bar U_i\}$ satisfying the conclusion with
$\hat F = N \setminus \{\emptyset\}$, therefore this cover is a good cover (note that properties 1
and 3 remain valid when iterating the construction).

Let $\{A_i\}\subset\{B_i\}$ be the covers representing $A$ and $B$,
respectively. Let $\{\bar B_i\}$ be the good cover obtained by the above
construction, and $\{\bar A_i\}\subset \{\bar B_i\}$ be the subcover
consisting of those sets that correspond to the sets of $\{A_i\}$. Clearly, the
cover $\{\bar A_i\}$ inherits from $\{\bar B_i\}$ the three  properties listed above.
Let $\BA=\bigcup_i A_i$, $\bar \BA=\bigcup_i \bar A_i$, and $\BB$,
$\bar \BB$ be defined analogously.
We have the following commutative diagram:
$$
\begin{CD}
A @>\subset >> B\\
@A\pr_{A} AA @A\pr_B AA\\
\Delta(\{\bar A_i\}) @>\subset >> \Delta(\{\bar B_i\})\\
@V\pr_{\bar\BA}VV @V\pr_{\bar \BB}VV\\
\bar \BA @>\subset >> \bar\BB \\
@A\subset AA @A\subset AA\\
\BA @>\subset>> \BB
\end{CD}
$$
Consider the induced maps in homology. All vertical maps then become
isomorphisms. The lower two vertical maps are isomorphisms by the property
3 of Lemma~\ref{l:ind_acyclic}. The other four vertical maps are
isomorphisms by property 1 of the same lemma, by the fact that $\{\bar A_i\}$
and $\{\bar B_i\}$ are good covers and by Lemma~\ref{l:delta_homotopy}.  
\end{proof}

\begin{figure}[h]
\centering
\includegraphics{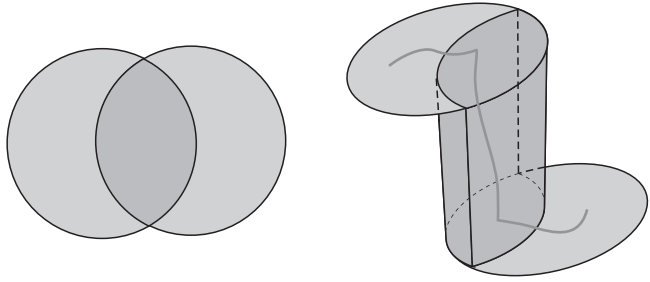}
\caption{Left: a cover $\{U_1,U_2\}$ consisting of two sets. Right: the space $\Delta(\{U_1,U_2\})$, and a gray path from the proof of Lemma~\ref{l:lift_loop}.}
\label{f:DeltaSpace}
\end{figure}

To prove Claim~\ref{c:pi1_notsurjective} we need the following lemma.

\begin{lemma}
\label{l:lift_loop}
Let $\{U_i\}$ be a cover such that all $U_i$ are  connected, and let $N(\{U_i\})$ be the nerve of $\{U_i\}$. For each $\gamma\in\pi_1(N(\{U_i\}))$, there is $\gamma '\in\pi_1(\Delta(\{U_i\}))$
such that $(\pr_{N(\{U_i\})})_* \gamma '=\gamma$.
\end{lemma}
\begin{proof}
Every element $\gamma\in\pi_1(N(\{U_i\}))$ can be realized by a loop, also denoted by $\gamma$, that belongs to the 1-skeleton of $N(\{U_i\})$. Such loop can be divided into pieces: each piece is an oriented 1-dimensional edge of $N(\{U_i\})$. For every such $e$, there is a path in $\Delta(\{U_i\})$, shown on Figure~\ref{f:DeltaSpace}, whose projection is $e$. Because every $U_i$ is connected, the paths in $\Delta(\{U_i\})$ can be joined together to form the loop $\gamma '$ whose projection is $\gamma$.
\end{proof}

\begin{proof}[Proof of Claim~\ref{c:pi1_notsurjective}]
We use notation from Section~\ref{s:fund}.
In particular,  $U=N(\{\BU_i\})$ is the nerve,
$\BU=\bigcup_i \BU_i$. (The sets of the cover are denoted by bold letters in Section~\ref{s:fund}.) Abusing notation, we will write $\Delta(U)$ instead of $\Delta(\{\BU_i\})$. We apply similar notation to $\{X_i\}$. 
Recall that $L\cap M=\BU$, $L=\BU\bigcup \BX$, and also that $\pi_1(U)=0$, $\pi_1(U\cup X)\neq 0$.
Consider the following commutative diagram:
$$
\begin{CD}
U @>\subset>>U\cup X\\
@A\pr_{U} AA @A\pr_{U\cup X} AA\\
\Delta(U)@>>>\Delta(U\cup X)\\
@V\pr_{\BU}V\sim V @V\pr_{\BU\cup \BX}V\sim V
\\
\BU @>\subset>>\BU\cup\BX
\end{CD}
$$
Take any element $\gamma\neq 0$ in $\pi_1(U\cup X)$.
By  Lemma~\ref{l:lift_loop}, find a loop $\gamma '$ in $\pi_1(\Delta(U\cup X))$
such that $\pr_{U\cup X}(\gamma')=\gamma$.
By Lemma~\ref{l:delta_homotopy}(a) the lower vertical maps are homotopy equivalences.
So if $\pi_1(U)\to \pi_1(U\cup X)$ were surjective, there would exist $\alpha\in \pi_1(\Delta(U))$ whose image under the inclusion-induced map $\pi_1(\Delta(U))\to\pi_1(\Delta(U\cup X))$ equals $\gamma'$.
By commutativity of the upper square of the diagram, we see that
$\gamma=0$ because $\pi_1(U)=0$, a contradiction.
\end{proof}


%
%
%
%


\end{document}